\documentclass[creativecommons,11pt]{eptcs}
\providecommand{\event}{EXPRESS'10}
\usepackage{breakurl}  
\usepackage{csp-cm,mytheorems,infrule,tech,available,yfonts}
\usepackage[all]{xy}

\newtheorem{theorem}{Theorem}
\newtheorem{example}[theorem]{Example}
\newtheorem{lemma}[theorem]{Lemma}

\newtheorem{definition}[theorem]{Definition}
\newtheorem{prop}[theorem]{Proposition}

\def\mathindent{1em}

\def\paragraph#1{\medskip\noindent\textbf{#1}\ }

\title{Models for CSP with availability information}
\author{Gavin Lowe} 
\def\titlerunning{Models for CSP with availability information}
\def\authorrunning{Gavin Lowe}

\begin{document}
\maketitle

\begin{abstract}
We consider models of CSP based on recording what events are available as
possible alternatives to the events that are actually performed.  We present
many different varieties of such models.  For each, we give a compositional
semantics, congruent to the operational semantics, and prove full abstraction
and no-junk results.  We compare the expressiveness of the different models.
\end{abstract}

\section{Introduction}

In this paper we consider a family of semantic models of CSP~\cite{awr:csp}
that record what events a process makes available as possible alternatives to
the events that are actually performed.  For example, the models will
distinguish $a \then STOP \extchoice b \then STOP$ and $a \then STOP
\intchoice b \then STOP$: the former offers its environment the choice
between~$a$ and~$b$, so can make~$a$ available before performing~$b$; however
the latter decides internally whether to offer~$a$ or~$b$, so cannot make~$a$
available before performing~$b$.

A common way of motivating process algebras (dating back to~\cite{milner80})
is to view a process as a black box with which the observer interacts.  The
models in this paper correspond to that black box having a light for each
event that turns on when the event is available (as in~\cite{vG01,vG93}); the
observer can record which lights turn on in addition to which events are
performed.

I initially became interested in such models by considering message-passing
concurrent programming languages that allow code to test whether a channel is
ready for communication without actually performing the communication.
In~\cite{gavin:ready}, I considered the effect of extending CSP with a
construct 
``$\If \Ready a \Then P \Else Q$''
that tests whether the event $a$ is ready for communication (i.e., whether
this process's environment is ready to perform~$a$), acting like~$P$ or~$Q$
appropriately.
The model in~\cite{gavin:ready} recorded what events were made available by a
process, in addition to the events actually performed.  We investigate such
models more fully in this paper.  We show that ---even without the
above construct--- there are many different variations, with different
expressive power.  

By convention, a denotational semantic model of CSP is always
\emph{compositional}, i.e., the semantics of a composite process is given in
terms of the semantics of its components.  Further, there are several other
desirable properties of semantic models:
\begin{description}
\item[Congruence to the operational semantics] The denotational semantics can
  either be extracted from the operational semantics, or calculated
  compositionally, both approaches giving the same result;

\item[Full abstraction]
The notion of semantic equivalence  corresponds to some natural
equivalence, typically defined in terms of testing;

\item[No-junk] The denotational semantic domain corresponds precisely to the
  semantics of processes: for each element of the semantic domain, we can
  construct a corresponding process.
\end{description}
Each of the semantic models in this paper satisfies these properties.

In Section~\ref{sec:offer} we describe our basic model.  We formalise the
notion of availability of events in terms of the standard operational
semantics.  We then formalise the denotational semantic domain, and explain
how to extract denotational information from the semantics.  We then give a
congruent compositional denotational semantics, and prove full abstraction and
no-junk results.

In Section~\ref{sec:variations} we describe variations on the basic model, in
two dimensions: one dimension restricts the number of observations of
availability between successive standard events; the other dimension allows
the simultaneous availability of multiple events to be recorded.  For each
resulting model, we describe compositional semantics, and full abstraction and
no-junk results (we omit some of the details because of lack of space, and to
avoid repetition).  We then study the relative expressive power of the models.

Finally, in Section~\ref{sec:conc}, we discuss various aspects of our models,
some additional potential models, and some related work. 

\paragraph{Overview of CSP}
We give here a brief overview of the syntax and semantics of CSP; for simplicity and
brevity, we consider a fragment of the language in this paper.  We also give a
brief overview of the Traces and Stable Failures Models of CSP\@.  For more
details, see~\cite{H85,awr:csp}.

CSP is a process algebra for describing programs or {\em processes}\/ that
interact with their environment by communication.  Processes communicate via
atomic events, from some set~$\Sigma$.  Events often involve passing values
over channels; for example, the event $c.3$ represents the value~$3$ being
passed on channel~$c$.  

The simplest process is $STOP$, which represents a deadlocked process that
cannot communicate with its environment.  The process $\Div$ represents a
divergent process that can only perform internal events.

The process $a \then P$ offers its environment the event~$a$; if the event is
performed, it then acts like~$P$.  The process $c?x \then P$ is initially
willing to input a value $x$ on channel~$c$, i.e.~it is willing to perform any
event of the form~$c.x$; it then acts like~$P$ (which may use~$x$).
Similarly, the process $?a:A \then P$ is initially willing to perform any
event~$a$ from~$A$; it then acts like~$P$ (which may use~$a$).

The process $P \extchoice Q$ can act like either $P$ or~$Q$, the choice being
made by the environment: the environment is offered the choice between the
initial events of~$P$ and~$Q$.  By contrast, $P \intchoice Q$ may act like
either~$P$ or~$Q$, with the choice being made internally, not under the
control of the environment;  $\Intchoice_{x:X} P_x$
nondeterministically acts like any~$P_x$ for $x$ in~$X$.
The process $P \timeout Q$ represents a sliding choice or timeout: it
initially acts like~$P$, but if no event is performed then it can internally
change state to act like~$Q$.

The process $P \parallel[A][B] Q$ runs $P$ and~$Q$ in parallel; $P$ is
restricted to performing events from~$A$; $Q$ is restricted to performing
events from~$B$; the two processes synchronise on events from~$A \inter B$.
The process $P \interleave Q$ interleaves $P$ and $Q$,
i.e.\ runs them in parallel with no synchronisation.

The process $P \hide A$ acts like~$P$, except the events from~$A$ are hidden,
i.e.~turned into internal, invisible events, denoted~$\tau$, which do not need
to synchronise with the environment.  
The process $P \rrn{R}$ represents $P$ where events are renamed according to
the relation~$R$, i.e., $P \rrn{R}$ can perform an event~$b$ whenever $P$ can
perform an event~$a$ such that $a~R~b$. 

Recursive processes may be defined equationally, or using the notation $\mu X
\spot P$, which represents a process that acts like $P$, where each occurrence
of $X$ represents a recursive instantiation of $\mu X \spot P$.

Prefixing ($\then$)  binds tighter than each of the binary
choice operators, which in turn bind tighter than the parallel operators.

CSP can be given both an operational and denotational semantics.  The
denotational semantics can either be extracted from the operational semantics,
or defined directly over the syntax of the language; see~\cite{awr:csp}.  It
is more common to use the denotational semantics when specifying or describing
the behaviours of processes, although most tools act on the operational
semantics.
A \emph{trace} of a process is a sequence of (visible) events that a process
can perform.  
If $tr$ is a
trace, then $tr \project A$ represents the restriction of~$tr$ to the events
in~$A$, whereas $tr \hide A$ represents $tr$ with the events from~$A$ removed;
concatenation is written ``$\,\cat\,$''; $A^*$ represents the set of traces
with events from~$A$.
A \emph{stable failure} of a process~$P$ is a pair $(tr,X)$, which represents
that $P$ can perform the trace~$tr$ to reach a stable state (i.e.~where no
internal events are possible) where $X$ can be refused, i.e., where none of
the events of~$X$ is available.  

\section{Availability information}
\label{sec:offer}

In this section we consider a model that record that particular events are
available during an execution.  We begin by extending the operational
semantics so as to formally define this notion of availability.  We then
define our semantic domain ---traces containing both standard events and
availability information--- with suitable healthiness conditions.  We then
present compositional trace semantics, and show that it is congruent to the
operational semantics.  Finally, we prove full abstraction and no-junk
results.

We write $\offer a$ to record that the event~$a$ is offered by a process,
i.e.~$a$ is available.  We augment the operational semantics with actions to
record such offers (we term these \emph{actions}, to distinguish them from
standard events).  Formally, we define a new transition relation $\transd{}$
from the standard transition relation~$\trans{}$
(see~\cite[Chapter~7]{awr:csp}) by:
%
\twoinpage{
\begin{eqnarray*}
P \transd{\alpha} Q & \;\iff\; & P \trans{\alpha} Q, 
  \gap \mbox{for $\alpha \in \Sigma \union \set{\tau}$},
\end{eqnarray*}}{
\begin{eqnarray*}
P \transd{\offer a} P & \;\iff\; & P \trans{a} \,.
\end{eqnarray*}}
For  example:
\(
a \then STOP \extchoice b \then STOP \;\transd{\offer a} \;
  a \then STOP \extchoice b \then STOP \;\transd{b}\; STOP.
\)
Note that the transitions corresponding to $\offer$ actions do not change the
state of the process.  

We now consider an appropriate form for the denotational semantics.  One might
wonder whether it is enough to record availability information only at the
\emph{end} of a trace (by analogy to the stable failures model).  However, a
bit of thought shows that such a model would be equivalent to the standard
Traces Model: a process can perform the trace $tr \cat \trace{\offer a}$
precisely if it can perform the standard trace~$tr \cat \trace{a}$.

We therefore record availability information \emph{throughout} the trace.
For convenience, for $A \subseteq \Sigma$, we define
\[
\offer A  =  \set{\offer a | a \in A}, \gap
A^\dagger  =  A \union \offer A, \gap
A^{\dagger\tau}  =  A^\dagger \union \set{\tau}.
\]
We define an \emph{availability trace} to be a sequence $tr$ in
$(\Sigma^\dagger)^*$.  We can extract the traces (of $\Sigma^\dagger$ actions)
from the operational semantics (following the approach in \cite[Chapter
  7]{awr:csp}):
\begin{definition}
\label{def:traces}
We write $P \TransTau{s} Q$, for $s = \trace{\alpha_1, \ldots, \alpha_n} \in
(\Sigma^{\dagger\tau})^*$, if there exist $P_0 = P,\linebreak[1] P_1,
\ldots,\linebreak[1] P_n = Q$ such that $P_i \transd{\alpha_{i+1}} P_{i+1}$
for $i = 0, \ldots, n-1$.
We write $P \Trans{tr} Q$, for $tr \in (\Sigma^\dagger)^*$, if there is some
$s$ such that $P \TransTau{s} Q$ and $tr = s \hide \tau$.
\end{definition}
\begin{example}
The process $a \then STOP \extchoice b \then STOP$ has the availability trace
$\trace{\offer a, b}$.  However, the process $a \then STOP \intchoice b \then
STOP$ does not have this trace.  This model therefore distinguishes these two
processes, unlike the standard Traces Model.
\end{example}

Note in particular that we may record the availability of events in unstable
states (where $\tau$ events are available), by contrast with models like the
Stable Failures Model that record (un)availability information only in stable
states.  The following example contrasts the two models. 
\begin{example}
\label{example:vs_failures}
The processes $a \then STOP$ and $a \then STOP \intchoice STOP$ are
distinguished in the Stable Failures Model, since the latter has stable
failure $(\trace{}, \set{a})$; however they have the same availability traces.

The processes $(a \then STOP \timeout b \then STOP) \intchoice (b \then STOP
\timeout a \then STOP)$ and $a \then STOP \intchoice b \then STOP$ are
distinguished in the Availability Traces Model, since only the former has the
availability trace $\trace{\offer a, b}$; however, they have the same stable
failures.
\end{example}

The availability-traces of process~$P$ are then $\set{tr | P \Trans{tr}}$.
%
The following definition captures the properties of this model.
\begin{definition}
\label{def:availability}
The \emph{Availability Traces Model}~$\A$ contains those sets~$T
\subseteq (\Sigma^\dagger)^*$ that satisfy the following conditions:
\begin{enumerate}
\item
$T$ is non-empty and prefix-closed.

\item
\label{healthy:offer_duplicate}
$\offer$  actions can always be remove from or duplicated
within a trace:
\begin{eqnarray*}
tr \cat \trace{\offer a} \cat tr' \in T & \implies &
  tr \cat \trace{\offer a, \offer a} \cat tr' \in T \land 
  tr \cat tr' \in T.
\end{eqnarray*}

\item
\label{healthy:offer_implies_event}
If a process can offer an event it can perform it:
\(
tr \cat \trace{\offer a} \in T \; \implies \; tr \cat \trace{a} \in T.
\)

\item
\label{healthy:event_implies_offer}
If a process can perform an event it can first offer it:
\(
tr \cat \trace{a} \cat tr' \in T \; \implies \;
  tr \cat \trace{\offer a, a} \cat tr' \in T.
\)
\end{enumerate}
\end{definition}
\begin{lemma}
\label{lem:opsem-healthy}
For all processes~$P$,\, $\set{tr | P \Trans{tr}}$ is an element of the
Availability Traces Model, i.e., satisfies the four healthiness conditions.
\end{lemma}
%


\paragraph{Compositional traces semantics}
We now give compositional rules for the traces of a process.  
We write $\Tracesa[P]$ for the traces of~$P$\footnote{We include the
subscript~``$A$'' in $\Tracesa[P]$ to distinguish this semantics from the
standard traces semantics, $\Traces[P]$.}.
Below we will show that these are congruent to
the operational definition above.

$STOP$ and $\Div$ are equivalent in this model: they can neither perform nor
offer standard events.  The process $a \then P$ can initially signal that it
is offering~$a$; it can 
then perform~$a$, and continue like~$P$.
\begin{eqnarray*}
\Tracesa[STOP] & = & \Tracesa[\Div] \;=\; \set{\trace{}}
\\
\Tracesa[a \then P] & = &
  \begin{align}
  \set{\offer a}^* \union 
   \set{ tr \cat \trace{a} \cat tr' | 
       tr \in \set{\offer a}^* \land tr' \in \Tracesa[P]}.
  \end{align}
\end{eqnarray*}

The process $P \timeout Q$ can either perform a trace of~$P$, or can perform a
trace of~$P$ with no standard events, and then (after the timeout) perform a
trace of~$Q$.  The process $P \intchoice Q$ can perform traces of either of
its components; the semantics of replicated nondeterministic choice is the
obvious generalisation.  
\begin{eqnarray*}
\Tracesa[P \timeout Q] & = & 
  \begin{align}
  \Tracesa[P] \union\null
  \set{ tr_P \cat tr_Q | 
     tr_P \in \Tracesa[P] \land tr_P \project \Sigma = \trace{}
     \land tr_Q \in \Tracesa[Q]},
  \end{align}
\\
\Tracesa[P \intchoice Q] & = & \Tracesa[P] \union \Tracesa[Q],
\\
\Tracesa[\Intchoice_{i \in I} P_i] & = & \Union\nolimits_{i \in I} \Tracesa[P_i].
\end{eqnarray*}

Before the first visible event, the process $P \extchoice Q$ can perform an
$\offer a$ action if \emph{either}~$P$ or~$Q$ can do so.  Let $tr \interleave
tr'$ be the set of ways of interleaving $tr$ and $tr'$ (this operator is
defined in~\cite[page~67]{awr:csp}).  The three sets in the definition below
correspond to the cases where (a)~neither process performs any visible events,
(b)~$P$~performs at least one visible event (after which, $Q$ is turned off),
and (c)~the symmetric case where $Q$~performs at least one visible event.
\begin{eqns}
\Tracesa[P \extchoice Q] = \\*
\gap
  \begin{align}
  \set{ tr | 
    \exists tr_P \in \Tracesa[P], tr_Q \in \Tracesa[Q] \spot 
       tr_P \project \Sigma = tr_Q \project \Sigma = \trace{} \land
    tr \in tr_P \interleave tr_Q} 
  \union\null \\
  \set{ tr \cat \trace{a} \cat tr_P' | 
    \exists tr_P \cat \trace{a} \cat tr_P' \in \Tracesa[P], 
        tr_Q \in \Tracesa[Q] \spot \\
    \gap tr_P \project \Sigma = tr_Q \project \Sigma = \trace{} \land
    a \in \Sigma \land tr \in tr_P \interleave tr_Q} 
  \union\null \\
  \set{ tr \cat \trace{a} \cat tr_Q' | 
    \exists tr_P \in \Tracesa[P], 
        tr_Q \cat \trace{a} \cat tr_Q' \in \Tracesa[Q] \spot \\
    \gap tr_P \project \Sigma = tr_Q \project \Sigma = \trace{} \land
    a \in \Sigma \land tr \in tr_P \interleave tr_Q} .
  \end{align}
\end{eqns}

In a parallel composition of the form $P \parallel[A][B] Q$,\, $P$ is
restricted to actions from $A^\dagger$, and $Q$ is restricted to actions from
$B^\dagger$.  Further, $P$ and $Q$ must synchronise upon both standard events
from $A \inter B$ and offers of events from $A \inter B$.  We write $tr_P
\parallel[(A \inter B)^\dagger] tr_Q$ for the set of ways of synchronising
$tr_P$ and $tr_Q$ on actions from $(A \inter B)^\dagger$ (this operator is
defined analogously to in~\cite[page~70]{awr:csp}). The semantics of
interleaving is similar. 
%
\begin{eqnarray*}
\Tracesa[{P \parallel[A][B] Q}] & = & 
  \begin{align}
  \set{tr | \exists tr_P \in \Tracesa[P] \inter (A^\dagger)^*, 
      tr_Q \in \Tracesa[Q] \inter (B^\dagger)^* \spot 
        tr \in tr_P \parallel[(A \inter B)^\dagger] tr_Q }.
  \end{align}
\\
\Tracesa[P \interleave Q] & = & 
   \set{tr | \exists tr_P \in \Tracesa[P],  tr_Q \in \Tracesa[Q] \spot 
     tr \in tr_P \interleave tr_Q }.
\end{eqnarray*}

The semantic equation for hiding of~$A$ captures that $\offer A$ actions are
blocked, and $A$ events are internalised.  For relational renaming, we lift the renaming to apply to $\offer$ actions,
i.e.~$(\offer a)~R~(\offer b)$ if and only if $a~R~b$; we then lift the
relation to traces by pointwise application.  The semantic equation is then a
further lift of~$R$. 
\begin{eqnarray*}
\Tracesa[P \hide A] & = &
  \set{ tr_P \hide A | 
    tr_P \in \Tracesa[P] \land tr_P \project \offer A = \trace{}}.
\\
\Tracesa[P\rrn{R}] & = &
  \set{tr | \exists tr_P \in \Tracesa[P] \spot tr_P~R~tr }.
\end{eqnarray*}

We now consider the semantics of recursion.  Our approach follows the standard
method using complete partial orders; see, for example,
\cite[Appendix~A.1]{awr:csp}.
\begin{lemma}
\label{lem:cpo}
The Availability Traces Model forms a complete partial order under the
subset ordering~$\subseteq$, with $\Tracesa[\Div]$ as the bottom element.
\end{lemma}
%
%
%
\begin{lemma}
Each of the operators is continuous with respect to the $\subseteq$ ordering.
\end{lemma}
Hence from Tarski's Theorem, each mapping~$F$ definable using the operators of
the language has a least fixed point given by $\Union_{n \ge 0}
F^n(\Div)$.  This justifies the following definition.
\begin{eqns}
\Tracesa[\mu X \spot F(X)] = 
  \mbox{the $\subseteq$-least fixed point of the
  semantic mapping corresponding to~$F$.}
\end{eqns}

\smallskip

The following theorem shows that the two ways of capturing the traces are
congruent; it can be proved by a straightforward structural induction.
\begin{theorem}
\label{thm:congruent}
For all traces~$tr \in (\Sigma^\dagger)^*$:\,
\(
tr \in \Tracesa[P]  \;\mbox{iff}\;  P \Trans{tr}.
\)
\end{theorem}
%
%
\begin{theorem}
\label{thm:healthy}
For all processes, $\Tracesa[P]$ is a member of the Availability Traces Model
(i.e., it satisfies the conditions of Definition~\ref{def:availability}).
\end{theorem}
%


\paragraph{Full abstraction}
%
We can show that this model is fully abstract with respect to a form of
testing in the style of~\cite{dNH84}.  We consider tests that may detect the
availability of events.  Following~\cite{gavin:ready}, we write $\Ready a \amp
T$ for a test that tests whether $a$ is available, and if so acts like the
test~$T$.  We also allow a test $SUCCESS$ that represents a successful test,
and a simple form of prefixing.  Formally, tests are defined by the grammar:
\begin{eqnarray*}
T & ::= & SUCCESS | a \then T | \Ready a \amp T .
\end{eqnarray*}
We consider testing systems comprising a test~$T$ and a process~$P$, denoted
$T \parallel P$.  We define the semantics of testing systems by the rules
below; $\omega$ indicates that the test has succeeded, and $\Omega$ represents
a terminated testing system.
\twoinpage{%
\begin{eqnarray*}
\hspace*{-3cm}SUCCESS \parallel P & \trans{\omega} & \Omega
\end{eqnarray*}
}{
\begin{infrule}
P \;\transd{\tau}\; Q
\derive
T \parallel P \;\trans{\tau}\; T \parallel Q
\end{infrule}}
\twoinpage{%
\begin{infrule}
P \;\transd{a}\; Q
\derive
a \then T \parallel P \;\trans{\tau}\; T \parallel Q
\end{infrule}
}{
\begin{infrule}
P \;\transd{\offer a}\; P
\derive
\Ready a \amp T \parallel P \;\trans{\tau}\; T \parallel P
\end{infrule}}
We say that $P$ may pass the test~$T$, denoted $P \may T$, if $T \parallel P$
can perform~$\omega$ (after zero or more $\tau$s). 

We now show that if two processes are denotationally different, we can produce
a test to distinguish them, i.e., such that one process passes the test, and
the other fails it.  Let $tr \in (\Sigma^\dagger)^*$.  We can construct a test
$T_{tr}$ that detects the trace~$tr$.
\begin{eqnarray*}
T_{\trace{}} & = & SUCCESS, \\
T_{\trace{a} \cat tr} & = & a \then T_{tr}, \\
T_{\trace{\offer a} \cat tr} & = & \Ready a \amp T_{tr}.
\end{eqnarray*}
The following lemma can be proved by a straightforward induction on the length
of~$tr$:
\begin{lemma}
For all  processes~$P$,\, $P \may T_{tr}$ if and only if $tr \in \Tracesa[P]$.
\end{lemma}
%

\begin{theorem}
$\Tracesa[P] = \Tracesa[Q]$ if and only if $P$ and $Q$ pass the same tests.
\end{theorem}
\begin{proof}
The only if direction is trivial.  If $\Tracesa[P] \ne \Tracesa[Q]$ then
without loss of generality suppose $tr \in \Tracesa[P] - \Tracesa[Q]$; then
${P \may T_{tr}}$ but not $Q \may T_{tr}$.
\end{proof}

We now show that the model contains no junk: each element of the model
corresponds to a process.
\begin{theorem}
\label{thm:no-junk}
Let $T$ be a member of the Availability Traces Model.  Then there is a
process~$P$ such that $\Tracesa[P] = T$.
\end{theorem}
\begin{proof}
Let $tr$ be a trace.  We can construct a process~$P_{tr}$ as follows:
\begin{eqnarray*}
P_{\trace{}} & = & STOP, \\
P_{\trace{a} \cat tr} & = & a \then P_{tr}, \\
P_{\trace{\offer a} \cat tr} & = & a \then \Div \timeout P_{tr}.
\end{eqnarray*}
Then the traces of $P_{tr}$ are just~$tr$ and those traces implied from~$tr$
by the healthiness conditions of Definition~\ref{def:availability}.  Formally,
we can prove this by induction on~$tr$.  For example:
\begin{itemize}
\item
The traces of $P_{\trace{a} \cat tr}$ are prefixes of traces of the
form\footnote{We write $(\offer a)^k$ to denote a trace containing $k$ copies
  of $\offer a$.}  $(\offer a)^k \cat \trace{a} \cat tr'$, where $k \ge 0$ and
$tr'$ is a trace of $P_{tr}$.  Hence (by the inductive hypothesis) $tr'$ is
implied from~$tr$ by the healthiness conditions.  Thus $\trace{a} \cat tr'$ is
implied from $\trace{a} \cat tr$.  Finally $(\offer a)^k \cat \trace{a} \cat
tr'$ is implied from $\trace{a} \cat tr$ by $k$ applications of healthiness
condition~\ref{healthy:event_implies_offer}.

\item
The traces of $P_{\trace{\offer a} \cat tr}$ are of two forms:
  \begin{itemize}
  \item
    Prefixes of traces of the form $(\offer a)^k \cat \trace{a}$, which is
    implied from $\trace{\offer a}$ by healthiness
    conditions~\ref{healthy:offer_duplicate}
    and~\ref{healthy:offer_implies_event}.

  \item
    Traces of the form $(\offer a)^k \cat tr'$ where $tr'$ is a trace
    of~$P_{tr}$.  Hence (by the inductive hypothesis) $tr'$ is implied
    from~$tr$ by the healthiness conditions.  And so $(\offer a)^k \cat tr'$
    is implied from~$\trace{\offer a} \cat tr$ by healthiness
    condition~\ref{healthy:offer_duplicate}.
  \end{itemize}
\end{itemize}

Then $P = \Intchoice_{tr \in T} P_{tr}$ is such that $\Tracesa[P] = T$.
\end{proof}


\section{Variations}
\label{sec:variations}

In this section we consider variations on the model of the previous section,
extending the models along essentially two different dimensions.  We first
consider models that place a limit on the number of $\offer$ actions between
consecutive standard events.  We then consider models that record the
availability of \emph{sets} of events.  Finally we combine these two
variations, to produce a hierarchy of different models with different
expressive power (illustrated in Figure~\ref{fig:hierarchy}).  For each
variant, we sketch how to adapt the semantic model and full abstraction result
from Section~\ref{sec:offer}.  We concentrate on discussing the relationship
between the different models.

\subsection{Bounded  availability actions}
\label{sec:bounded-number}

Up to now, we have allowed arbitrarily many $\offer$ actions between
consecutive standard events.  It turns out that we can restrict this.  For
example, we could allow at most one $\offer$ action between consecutive
standard events (or before the first event, or after the last event).  This
model is more abstract than the previous; for example, it identifies the
processes
\[
\begin{align}
(a \then STOP \extchoice b \then STOP) 
\intchoice (a \then STOP \extchoice c \then STOP) 
\intchoice (b \then STOP \extchoice c \then STOP)
\end{align}
\]
and 
\[
(a \then STOP \extchoice b \then STOP \extchoice c \then STOP),
\]
whereas the previous model distinguished them by the trace $\trace{\offer
  a,\linebreak[1] \offer b, c}$.

More generally, we define the model that allows at most $n$ $\offer$ actions
between consecutive standard events.  Let $Obs_n$ be the set of availability
traces with this property.  Then the model $\A_n$ is the restriction of~$\A$
to $Obs_n$, i.e., writing $\Tracesan{n}$ for the semantic function for $\A_n$,
we have $\Tracesan{n}[P] = \Tracesa[P] \inter Obs_n$.  In particular, $\A_0$
is equivalent to the standard traces model.

%
%

The following example shows that the models become strictly more refined as
$n$ increases; further, the full Availability Traces Model $\A$ is finer than
each of the approximations~$\A_n$.
\begin{example}
\label{example:models-monotonic-n}
Consider the processes
\begin{eqnarray*}
P_0 & = & STOP \\
P_{n+1} & = &  
  (a \then STOP \intchoice  b \then STOP) \timeout P_n.
\end{eqnarray*}
Suppose $n$ is non-zero and even (the case of odd $n$ is similar).  Processes
$P_n$ and~$P_{n+1}$ can be distinguished in model~$\A_n$ and
model~$\A$, since only $P_{n+1}$ has the trace $\trace{\offer
  a,\linebreak[1] \offer b,\linebreak[1] \offer a,\linebreak[1] \offer
  b,\linebreak[1] \ldots, \offer a,\linebreak[1] \offer b, a}$ with $n$
$\offer$ actions.  However, these processes are equal in
model~$\A_{n-1}$.
\end{example}

Following Roscoe~\cite{awr:revivals}, we write $M \preceq M'$ if model~$M'$ is
finer (i.e.~distinguishes  more processes) than model~$M$, and $\prec$ for the
corresponding strict relation.
The above example shows
\[
\A_0 \equiv \mathcal{T}\prec \A_1  \prec \A_2 \prec \ldots \prec 
  \A.
\]

It is easy to see that these models are all compositional: in all the semantic
equations, the presence of a trace from $Obs_n$ in a composite process is
always implied by the presence of traces from $Obs_n$ in the subcomponents.
It is important, here, that the number of consecutive offers is
downwards-closed: the same result would not hold if we considered a model that
includes \emph{exactly} $n$ offer actions between successive standard events,
for in an interleaving $P \interleave Q$, a sequence of $n$ consecutive offers
may be formed from $k$~offers of~$P$ and $n-k$ offers of~$Q$.

In some cases, the semantic equations have to be adapted slightly to ensure the
traces produced are indeed from $Obs_n$, for example:
\begin{eqns}
\Tracesan{n}[{P \parallel[A][B] Q}] = \\*
\gap
  \begin{align}
  \set{tr | \exists tr_P \in \Tracesan{n}[P] \inter (A^\dagger)^*, 
      tr_Q \in \Tracesan{n}[Q] \inter (B^\dagger)^* \spot 
        tr \in ( tr_P \parallel[(A \inter B)^\dagger] tr_Q ) \inter Obs_n }.
  \end{align}
\end{eqns}

The healthiness conditions need to be adapted slightly to reflect that only
traces from $Obs_n$ are included.  For example,
condition~\ref{healthy:event_implies_offer} becomes
\begin{enumerate}
\item[\ref{healthy:event_implies_offer}$'$.]
\(
\begin{align}
tr \cat \trace{a} \cat tr' \in T 
\land tr \cat \trace{\offer a, a} \cat tr' \in Obs_n \implies 
   tr \cat \trace{\offer a, a} \cat tr' \in T.
\end{align}
\)
\end{enumerate}

Finally, the full abstraction result still holds, but the tests need to be
restricted to include at most $n$ successive $\Ready$ tests.  And the no-junk
result still holds. 


\subsection{Availability sets}
\label{sec:availSets}

The models we have considered so far have considered the availability of a
\emph{single} event at a time.  If we consider the availability of a
\emph{set} of events, can we distinguish more processes?  The answer turns out
to be yes, but only with processes that can either diverge or that exhibit
unbounded nondeterminism (a result which was surprising to me). 

We will consider actions of the form $\offer A$, where $A$ is
a set of events, representing that all the events in~$A$ are simultaneously
available.  We can adapt the derived operational semantics appropriately:
\begin{eqnarray*}
P \transd{\alpha}_\power Q & \;\iff\; & P \trans{\alpha} Q, 
  \gap \mbox{for $\alpha \in \Sigma \union \set{\tau}$}, \\
P \transd{\offer A}_\power P & \;\iff\; & \forall a \in A \spot P \trans{a} \,.
\end{eqnarray*}

For convenience, we define
\begin{eqnarray*}
A^{\power\dagger} & = & A \union \set{\offer B | B \in \power A}.
\end{eqnarray*}
Traces will then be from $(\Sigma^{\power\dagger})^*$.  We can extract traces
of this form from the derived operational semantics as in
Definition~\ref{def:traces} (writing $\TransTau{tr}_\power$ and
$\Trans{tr}_\power$ for the corresponding relations).

We call this model the Availability Sets Traces Model, and will sometimes refer
to the previous model as the Singleton Availability Traces Model, in order to
emphasise the difference.
\begin{definition}
The \emph{Availability Sets Traces Model} $\A^{\power}$ contains
those sets ${T \subseteq (\Sigma^{\power\dagger})^*}$ that satisfy the
following conditions.
\begin{enumerate}
\item
$T$ is non-empty and prefix-closed.

\item
\label{healthy:offer_duplicate-sets}
$\offer$  actions can always be removed from or duplicated
within a trace:
\begin{eqnarray*}
tr \cat \trace{\offer A} \cat tr' \in T & \implies &
  tr \cat \trace{\offer A, \offer A} \cat tr' \in T \land 
  tr \cat tr' \in T.
\end{eqnarray*}

\item
\label{healthy:offer_implies_event-sets}
If a process can offer an event it can perform it:
\(
tr \cat \trace{\offer A} \in T \;\implies \;
  \forall a \in A \spot tr \cat \trace{a} \in T.
\)

\item
\label{healthy:event_implies_offer-sets}
If a process can perform an event it can first offer it:
\begin{eqnarray*}
tr \cat \trace{a} \cat tr' \in T & \implies &
  tr \cat \trace{\offer \set{a}, a} \cat tr' \in T.
\end{eqnarray*}

\item
The offers of a process are subset-closed
\begin{eqnarray*}
tr \cat \trace{\offer A} \cat tr'  \in T \land B \subseteq A & \implies &
  tr \cat \trace{\offer B} \cat tr'  \in T.
\end{eqnarray*}

\item
Processes can always offer the empty set
\(
tr \cat  tr'  \in T \; \implies \;
  tr \cat \trace{\offer \set{}} \cat tr'  \in T.
\)
\end{enumerate}
\end{definition}
%
%
\begin{lemma}
\label{lem:opsem-healthy-sets}
For all processes~$P$,\,  $\set{tr | P \Trans{tr}_\power }$ is an
element of the Availability Sets Traces Model.
\end{lemma}
%


\paragraph{Compositional semantics}
We give below semantic equations for the Availability Sets Traces Model.  Most
of the clauses are straightforward adaptations of the corresponding clauses in
the Singleton Availability Traces Model.

For the parallel operators and external choice, we define an operator
$\parallel[X]^{\power}$ such that $tr \parallel[X]^{\power} tr'$ gives all
traces resulting from traces~$tr$ and~$tr'$, synchronising on events and
offers of events from~$X$.  The definition is omitted due to space
restrictions. 

For relational renaming, we lift the renaming to apply to $\offer$ actions, by
forming the subset-closure of the relational image:
\begin{eqnarray*}
(\offer A)~R~(\offer B) & \iff & 
  \forall b \in B \spot \exists a \in A \spot a~R~b.
\end{eqnarray*}
We again lift it to traces pointwise.  

The semantic clauses are as follows.
\begin{eqns}
\Tracesap[STOP]  = \Tracesap[\Div] = (\offer \set{})^* ,
\nl
\Tracesap[a \then P] = 
  \begin{align}
  Init \union 
   \set{ tr \cat \trace{a} \cat tr' | 
       tr \in Init \land tr' \in \Tracesap[P]}, \\
  \mbox{where } Init = \set{\offer \set{}, \offer \set{a}}^* ,
  \end{align}
\nl
\Tracesap[P \timeout Q] = 
  \begin{align}
  \Tracesap[P] \union\null
  \set{ tr_P \cat tr_Q | 
     tr_P \in \Tracesap[P] \land tr_P \project \Sigma = \trace{}
     \land tr_Q \in \Tracesap[Q]},
  \end{align}
\nl
\Tracesap[P \intchoice Q] = 
  \Tracesap[P] \union \Tracesap[Q], 
\nl
\Tracesap[P \extchoice Q] = \\*
\gap
  \begin{align}
  \set{ tr | 
    \exists tr_P \in \Tracesap[P], tr_Q \in \Tracesap[Q] \spot 
       tr_P \project \Sigma = tr_Q \project \Sigma = \trace{} \land
    tr \in tr_P \parallel[\set{}]^{\power} tr_Q} 
  \union\null \\
  \set{ tr \cat \trace{a} \cat tr_P' | 
    \exists tr_P \cat \trace{a} \cat tr_P' \in \Tracesap[P], 
        tr_Q \in \Tracesap[Q] \spot \\
    \gap tr_P \project \Sigma = tr_Q \project \Sigma = \trace{} \land
    a \in \Sigma \land tr \in tr_P \parallel[\set{}]^{\power} tr_Q} 
  \union\null \\
  \set{ tr \cat \trace{a} \cat tr_Q' | 
    \exists tr_P \in \Tracesap[P], 
        tr_Q \cat \trace{a} \cat tr_Q' \in \Tracesap[Q] \spot \\
    \gap tr_P \project \Sigma = tr_Q \project \Sigma = \trace{} \land
    a \in \Sigma \land tr \in tr_P \parallel[\set{}]^{\power} tr_Q} ,
  \end{align}
\nl
\Tracesap[{P \parallel[A][B] Q}] = 
  \begin{align}
  \set{tr | \exists tr_P \in \Tracesap[P] \inter (A^{\power\dagger})^*, 
      tr_Q \in \Tracesap[Q] \inter (B^{\power\dagger})^* \spot
        tr \in tr_P \parallel[A \inter B]^{\power} tr_Q },
  \end{align}
\nl
\Tracesap[P \interleave Q] = 
  \set{tr | \exists tr_P \in \Tracesap[P],  tr_Q \in \Tracesap[Q] \spot 
     tr \in tr_P \parallel[\set{}]^{\power} tr_Q },
\nl
\Tracesap[P \hide A] = 
  \set{ tr_P \hide A | 
    tr_P \in \Tracesap[P] 
    \land \forall X \spot \offer X \IN tr_P \implies X \inter A = \set{}},
\nl
\Tracesap[P\rrn{R}] =
  \set{tr | \exists tr_P \in \Tracesap[P] \spot tr_P~R~tr },
\nl
\Tracesap[\mu X \spot F(X)] = 
  \mbox{the $\subseteq$-least fixed point of the
  semantic mapping corresponding to~$F$.}
\end{eqns}

\begin{theorem}
The semantics is congruent to the operational semantics:
\(
tr \in \Tracesap[P] \;\mbox{iff}\; P \Trans{tr}_\power.
\)
\end{theorem}


\paragraph{Full abstraction}
In order to prove a full abstraction result, we extend our class of tests to
include a test of the form $\Ready A \amp P$, which tests whether all the
events in~$A$ are available, and if so acts like the test~$T$.  Formally, this
test is captured by the following rule.
\begin{infrule}
P \transd{\offer A} P
\derive
\Ready A \amp T \parallel P \trans{\tau} T \parallel P
\end{infrule}
Given $tr \in (\Sigma^{\power\dagger})^*$, we can construct a test $T_{tr}$
that detects the trace~$tr$ as follows.
\begin{eqnarray*}
T_{\trace{}} & = & SUCCESS \\
T_{\trace{a} \cat tr} & = & a \then T_{tr} \\
T_{\trace{\offer A} \cat tr} & = & \Ready A \amp T_{tr}
\end{eqnarray*}
The full abstraction proof then proceeds precisely as in
Section~\ref{sec:offer}.

We can prove a no-junk result as in Section~\ref{sec:offer}.  Given
trace~$tr$, we can  construct a process~$P_{tr}$ as follows:
\begin{eqnarray*}
P_{\trace{}} & = & STOP, \\
P_{\trace{a} \cat tr} & = & a \then P_{tr}, \\
P_{\trace{\offer A} \cat tr} & = & (?a:A \then \Div) \timeout P_{tr}.
\end{eqnarray*}
Then the traces of $P_{tr}$ are just~$tr$ and those traces implied from~$tr$
by the healthiness conditions.  Again, given an element $T$ from the
Availability Sets Traces Model, we can define $P = \Intchoice_{tr \in T}
P_{tr}$; then $\Tracesap[P] = T$.


\paragraph{Distinguishing power}
%
We now consider the extent to which the Availability Sets Model can
distinguish processes that the Singleton Availability Model can't. 
\begin{example}
\label{example:sets_vs_singles}
The Availability Sets Traces Model distinguishes the processes
\begin{eqnarray*}
P & = & a \then STOP \extchoice b \then STOP, \\
Q & = & (a \then STOP \intchoice b \then STOP) \timeout Q, 
\end{eqnarray*}
since just $P$ has the trace $\trace{ \offer \set{a,b} }$.  However, these are
equivalent in the Singleton Availability Traces Model; in particular, both can
perform 
arbitrary sequences of $\offer a$ and $\offer b$ actions initially.
\end{example}
The process~$Q$ above can diverge (i.e., perform an infinite number of
internal~$\tau$ events corresponding to
timeouts).  We can obtain a similar effect without divergence, but using
unbounded nondeterminism.  
\begin{example}
\label{example:sets_vs_singles_2}
Consider 
\begin{eqnarray*}
Q_0 & = & STOP, \\
Q_{n+1} & = & (a \then STOP \intchoice b \then STOP) \timeout Q_n, \\
Q' & = & \Intchoice_{n \in \nat} Q_n.
\end{eqnarray*}
Then $P$ (from the previous example) and $Q'$  are distinguished in the
Availability Sets Traces Model but not the Singleton Availability Traces
Model.
\end{example}

For finitely nondeterministic, non-divergent processes, it is enough to
consider the availability of only \emph{finite} sets, since such a process can
offer an infinite set~$A$ iff and only if it can offer all its finite subsets.
However, for infinitely nondeterministic processes, one can make more
distinctions by considering infinite sets.
\begin{example}
\label{example:finite_vs_infinite}
Let $A$ be an infinite set of events.  Consider the processes
\[
?a : A \then STOP
\gap\mbox{and}\gap
\Intchoice_{b \in A} ?a: A - \set{b} \then STOP
\]
Then these have the same finite availability sets, but just the former has
all of~$A$ available.
\end{example}

\begin{prop}
\label{prop:singleton_vs_sets}
If $P$ and~$Q$ are non-divergent, finitely nondeterministic processes, that
are equivalent in the Singleton Availability Model, then they are equivalent
in the Availability Sets Model.
\end{prop}
%
%
\begin{proof}
Suppose, for a contradiction, that $P$ and $Q$ are non-divergent and finitely
deterministic, are equivalent in the Singleton Availability Model, but are
distinguished in the Availability Set Model.  Then, without loss of
generality, there are traces~$tr$ and $tr'$, and set of events~$A$ such that
$tr \cat \trace{ \offer A} \cat tr'$ is a trace of~$P$ but not of~$Q$.  By the
discussion in the previous paragraph, we may assume, without loss of
generality, that $A$ is finite, say $A = \set{a_1, \ldots, a_n}$.  Since $Q$
is non-divergent and finitely-nondeterministic, there is some bound, $k$~say,
on the number of consecutive $\tau$ events that it can perform after~$tr$.
Since $P$ can offer all of~$A$ after $tr$, it can also offer any individual
events from~$A$, sequentially, in an arbitrary order.  In particular, it has
the singleton availability trace
\[
tr \cat \trace{ \offer a_1, \ldots, \offer a_n }^{k+1} \cat tr'.
\]
Since $P$ and $Q$ are, by assumption, equivalent in the Singleton Availability
model, $Q$ also has this trace.  $Q$~must perform at most $k$ $\tau$ events
within the sub-trace $\trace{ \offer a_1, \ldots, \offer a_n }^{k+1}$.  This
tells us that there is a sub-trace within that, of length~$n$, containing no
$\tau$~events. Within this sub-trace there are no state changes (i.e., there
are only self-loops corresponding to the $\offer$ actions), and so all the
$a_i$ are offered in the same state.  Hence $tr \cat \trace{
  \offer A } \cat tr'$ is an availability set trace of~$Q$, giving a
contradiction.
\end{proof}


\paragraph{Bounded sets}
%
We can consider some variants on the Availability Sets Traces Model.

First, let us consider the model $\A^k$ that places a limit of size~$k$ upon
availability sets.  It is reasonable straightforward to produce 
compositional semantics for such models, and to adapt the full abstraction and
no-junk results.   It is perhaps surprising that such a semantics is
compositional, since a similar result does not hold for stable
failures~\cite{BL04} (although it is conjectured in~\cite{awr:revivals} that
this does hold for acceptances).

Clearly, $\A^1 \equiv \A$, and
$\A^0 \equiv \mathcal{T}$ (the standard traces model).
Examples \ref{example:sets_vs_singles} and \ref{example:sets_vs_singles_2}
show that $\A^2$ is finer that $\A^1$.
%
%
We can generalise those examples to show that each model $\A^k$ is
finer than $\A^{k-1}$. 
\begin{example}
\label{example:monotonic-k}
Let $A_k$ be a set of size~$k$.  Consider
\begin{eqnarray*}
P_k & = & ?a:A_k \then STOP , \\
Q_k & = & \Intchoice_{b \in A_k} ( ?a:A_k-\set{b} \then STOP ) \timeout Q_k.
\end{eqnarray*}
Then $P_k$ and~$Q_k$ are distinguished in $\A^k$ since only $P_k$ has
the trace $\trace{ \offer A_k }$.  However they are equivalent in
$\A^{k-1}$: in particular, both can initially perform any trace of
offers of size~$k-1$. 
\end{example}

The limit of the models~$\A^k$ considers arbitrary \emph{finite}
availability sets; we term this $\A^{\bbold{F}}$.  The model
$\A^{\bbold{F}}$ distinguishes the processes $P_k$ and $Q_k$ from
Example~\ref{example:monotonic-k}, for all~$k$, so is finer than each of the
models with bounded availability sets.  As shown by
Example~\ref{example:finite_vs_infinite}, $\A^{\bbold{F}}$ is coarser
than $\A^{\power}$.

In fact, for an arbitrary infinite cardinal $\kappa$, we can consider the
model~$\A^\kappa$ that places a limit of size~$\kappa$ upon availability sets.
Example~\ref{example:finite_vs_infinite} showed that considering finite
availability sets distinguishes fewer processes than allowing infinite
availability sets, i.e.~$\A^{\bbold{F}} \prec \A^\kappa$.  The following
example shows that the models become finer as $\kappa$ increases.
\begin{example}
\label{example:monotonic-infinite-kappa}
Pick an infinite cardinal $\kappa$, and pick alphabet $\Sigma$ such that
$card(\Sigma) \ge \kappa$.  Then the processes
\begin{eqnarray*}
P_\kappa & = & \Intchoice_{A \subseteq \Sigma,\, card(A) = \kappa} ?a : A \then STOP, \\
Q_\kappa & = & \Intchoice_{A \subseteq \Sigma,\, card(A) < \kappa} ?a : A \then STOP
\end{eqnarray*}
are distinguished by the model $\A^\kappa$, since only $P_\kappa$ can offer
sets of size~$\kappa$.  However, for $\lambda < \kappa$, they are not
distinguished by the model $\A^{\lambda}$; for example, if $P_\kappa$ has the
trace $\trace{ \offer A_1, \ldots, \offer A_n }$ in~$\A^{\lambda}$, then
$card(A_i) \le \lambda < \kappa$, for each~$i$; but also $A = \Union_{i = 1}^n
A_i$ has $card(A) \le \lambda < \kappa$, so $Q_\kappa$ can perform this trace
by picking~$A$ in the nondeterministic choice.
\end{example}
(In fact, this example shows that these models ---like the cardinals--- form a
proper class, rather than a set!)  In most applications, the alphabet~$\Sigma$
is countable; these models then coincide for processes with such an alphabet.
The model $\A^{\power}$ distinguishes the processes $P_\kappa$ and
$Q_\kappa$ from Example~\ref{example:monotonic-infinite-kappa}, for
all~$\kappa$, so is finer than each of the models $\A^\kappa$.

Summarising:
\[
\A^0 \equiv \mathcal{T} \prec \A^1 \equiv \A \prec 
  \A^2 \prec \ldots \prec \A^{\bbold{F}} \prec 
  \A^{\aleph_0} \prec \A^{\aleph_1} \prec \ldots \prec
  \A^{\power} .
\]

\subsection{Combining the variations}

We can combine the ideas from Sections~\ref{sec:bounded-number}
and~\ref{sec:availSets} to produce a family of models $\A_\n^\k$, where:
\begin{itemize}
\item 
$\k$ is either a natural number~$k$ or infinite cardinal~$\kappa$, indicating
  an upper bound on the size of availability sets, or the symbol $\bbold{F}$
  indicating arbitrary finite availability sets are allowed, or the symbol
  $\bbold{P}$ indicating arbitrary availability sets are allowed;

\item
$\n$ is either a natural number~$n$, indicating an upper bound on the number
  of availability sets between successive standard events, or the symbol
  $\bbold{F}$ indicating any finite number is allowed.
\end{itemize}

\begin{figure}[htbp]
\begin{center}
\ 
\xymatrix @R=8.5mm @C=14.85mm @!0{
& & & & & & & & & *[c]{\A^{\power} \equiv \A_\finite^{\power}} \\
\\
& & & & & & & *[c]{\A^{\aleph_2} \equiv \A_\finite^{\aleph_2}} \ar@{.}[uurr] & & 
  \A_3^{\power} \ar@{.}[uu] 
\\
& & & & & & *[c]{\A^{\aleph_1} \equiv \A_\finite^{\aleph_1}} \ar@{-}[ur] & & & 
  \A_2^{\power} \ar@{-}[u] 
\\
& & & & & *[c]{\A^{\aleph_0} \equiv \A_\finite^{\aleph_0}} \ar@{-}[ur] & & 
  \A_3^{\aleph_2} \ar@{.}[uu] \ar@{.}[uurr] & & \A_1^{\power} \ar@{-}[u]
\\
& & & & *[c]{\A^\finite \equiv \A_\finite^\finite} \ar@{-}[ur] & & 
  \A_3^{\aleph_1} \ar@{.}[uu] \ar@{-}[ur] &
  \A_2^{\aleph_2} \ar@{-}[u] \ar@{.}[uurr] 
\\
& & & & & \A_3^{\aleph_0} \ar@{.}[uu] \ar@{-}[ur] & 
  \A_2^{\aleph_1} \ar@{-}[u] \ar@{-}[ur] & \A_1^{\aleph_2} \ar@{-}[u]
  \ar@{.}[uurr]
\\
& & *[c]{\A^3 \equiv \A_\finite^3} \ar@{.}[uurr] & & 
  \A_3^\finite \ar@{.}[uu] \ar@{-}[ur] & 
  \A_2^{\aleph_0} \ar@{-}[u] \ar@{-}[ur] & \A_1^{\aleph_1} \ar@{-}[u]
  \ar@{-}[ur] 
\\
& *[c]{\A^2 \equiv \A_\finite^2} \ar@{-}[ur] & & & 
  \A_2^{\finite} \ar@{-}[u] \ar@{-}[ur] &  \A_1^{\aleph_0} \ar@{-}[u] \ar@{-}[ur] 
\\
{\A \equiv \A^1 \equiv \A_\finite^1\hspace{5mm}} \ar@{-}[ur] & & 
  \A_3^3 \ar@{.}[uu] \ar@{.}[uurr] & & 
  \A_1^{\finite} \ar@{-}[u] \ar@{-}[ur] \\
& \A_3^2 \ar@{.}[uu] \ar@{-}[ur] & \A_2^3 \ar@{-}[u] \ar@{.}[uurr] & 
\\
{\A_3 \equiv \A_3^1} \ar@{.}[uu] \ar@{-}[ur] & \A_2^2 \ar@{-}[u] \ar@{-}[ur] & 
   \A_1^3 \ar@{-}[u] \ar@{.}[uurr] 
\\
\A_2 \equiv \A_2^1 \ar@{-}[u] \ar@{-}[ur] & \A_1^2 \ar@{-}[u] \ar@{-}[ur] 
\\
\A_1 \equiv \A_1^1 \ar@{-}[u] \ar@{-}[ur] 
\\
{\A_0 \equiv \A_\n^0 \equiv \A_0^\k \equiv \mathcal{T}} \ar@{-}[u]
}\ 
\end{center}
\caption{The hierarchy of models\label{fig:hierarchy}}
\end{figure}
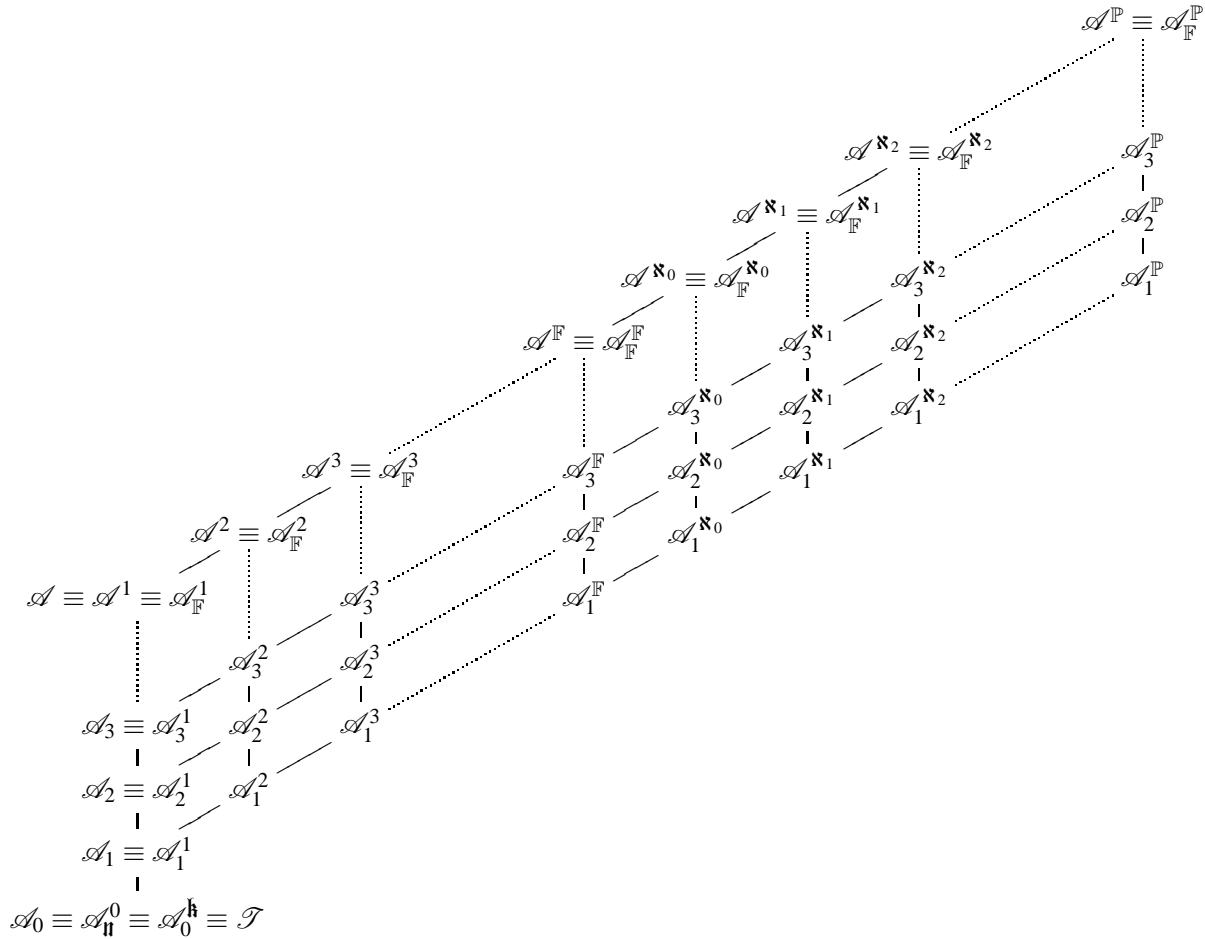

If $\k=0$ or $\n=0$ then $\A_\n^\k$ is just the standard traces model.  Further,
$\A_\n \equiv \A_\n^1$ and $\A^\k \equiv \A_{\bbold{F}}^\k$.

We can show that this family is ordered as the natural extension of the
earlier (one-parameter) families; the relationship between the models is
illustrated in Figure~\ref{fig:hierarchy}.  In particular, these models are
distinct for $\n, \k \ne 0$.  We can re-use several of the earlier examples to
this end.
Example~\ref{example:models-monotonic-n} shows that for each~$\k \ne 0$
\[
\A_0^\k \prec \A_1^\k  \prec \A_2^\k \prec \ldots \prec 
  \A_{\bbold{F}}^\k.
\]
The following example generalises Example~\ref{example:monotonic-k}.
\begin{example}
Let $n$ and $k$ be positive natural numbers.  Let $A$ be a set of size
$n\times k + 1$.    Consider
\begin{eqnarray*}
P & = & ?a:A \then STOP , \\
Q & = & \Intchoice_{B \subset A}  ?a:B \then STOP.
\end{eqnarray*}
Let $A_1, \ldots, A_n, \set{a}$ be a partition of~$A$, where each $A_i$ is of
size~$k$.  Then $\trace{\offer A_1, \ldots, \offer A_n, a}$ is a trace of~$P$
but not of~$Q$, so these processes are distinguished by~$\A_n^k$.
However, the two processes are equivalent in~$\A_n^{k-1}$.  Hence
$\A_n^{k-1} \prec \A_n^k$.  Further,
$\A_n^{\bbold{F}}$ distinguishes these processes, for all
(finite)~$n$ and~$k$, so $\A_n^k \prec \A_n^{\bbold{F}}$.
\end{example}
Example~\ref{example:finite_vs_infinite} shows that
$\A_\n^{\bbold{F}} \prec \A_\n^\kappa$ for each infinite
cardinal~$\kappa$ and for each~$\n$.  Further,
Example~\ref{example:monotonic-infinite-kappa} shows that if $\lambda <
\kappa$ are two infinite cardinals, then $\A_\n^{\lambda} \prec
\A_\n^\kappa \prec \A_\n^{\power}$.




\section{Discussion}
\label{sec:conc}

\paragraph{Simulation and model checking}
The models described in this paper are not supported by the model checker
FDR~\cite{awrFDR, FDRmanual}.  However, it is possible to simulate the
semantics, using a fresh event $\Offer.A$ to simulate the action
$\offer A$.  For example,  $P = a \then STOP \extchoice b \then
STOP$ would be simulated by
\begin{eqnarray*}
P_{sim} & = & a \then STOP_{sim} \extchoice b \then STOP_{sim}
  \extchoice \Offer?A:\power(\set{a,b}) \then P_{sim}, \\
STOP_{sim} & = & \Offer.\set{} \then STOP_{sim} .
\end{eqnarray*}

This simulation process, then, has the same traces as the original
process in the Availability Sets Model, but with each $\offer A$ action
replaced by $\Offer.A$.  The semantics in each of the other models can be
obtained by restricting the size or number of $\Offer$ events. 

In~\cite{awr:expressiveness}, Roscoe shows that any operational semantics that
is CSP-like, in a certain sense, can be simulated using standard CSP
operators.  One can define the operational semantics of the corrent paper in a
way that makes them
CSP-like, in this sense.   Roscoe's simulation is supported by a tool by
Gibson-Robinson~\cite{tom-simulation}, which has been used to automate the
simulation of the Singleton Availability Model and Availability Sets Model.
This opens up the possibility of using FDR to perform analyses in these
models.



\paragraph{Related and further models}
%
%
In~\cite{awr:revivals}, Roscoe investigates the hierarchy of finite linear
observation models of CSP\@.  All of these models record availability or
unavailability of events only in \emph{stable} states (if at all), unlike the
models of this paper.  Example~\ref{example:vs_failures} shows that the
Singleton Availability Model is incomparable with the Stable Failures Model.
In fact, this example shows that all of the models in this paper except the
Traces Model are incomparable with all of the models in Roscoe's hierarchy
except the Traces Model (so including the Ready Trace Model~\cite{OH83} and
the Refusal Testing Model~\cite{abida-thesis}); it is, perhaps, surprising
that the hierarchies are so unrelated.


We believe that we could easily adapt our models to extend any of the finite
linear observation models from~\cite{awr:revivals}, so as to obtain a
hierarchy similar to that in Figure~\ref{fig:hierarchy}: in effect, the
consideration of availability information is orthogonal to the finite linear
observations hierarchy.
Further, we have not considered divergences within this paper.  We believe
that it would be straightforward to extend this work with divergences, either
building models that are divergence-strict (like the traditional
Failures-Divergences Model~\cite{H85,awr:csp}), or non-divergence-strict (like
the model in~\cite{awr:seeing-beyond-divergence}).

\def\simulates{\stackrel{\subset}{\rightarrow}}

In~\cite{vG01,vG93}, van Glabbeek considers a hierarchy of different semantic
models in the linear time--branching time spectrum.  Several of the models
correspond to standard finite linear observation models, discussed above.  One
other model of interest is simulation.
\begin{definition}\cite{vG01}
A \emph{simulation} is a binary relation $R$ on processes such that for all
events~$a$, if $P~R~Q$ and $P \trans{a} P'$, then for some~$Q'$,\, $Q
\trans{a} Q'$ and $P'~R~Q'$.  
Process~$P$ can be simulated by $Q$, denoted $P \simulates Q$ if there is a
simulation~$R$ with $P~R~Q$.  $P$ and~$Q$ are \emph{similar} if $P \simulates
Q$ and 
$Q \simulates P$. 
\end{definition}
If $P \simulates Q$ and $P \TransTau{tr}_\power$ then one can show that $Q
\TransTau{tr}_\power$, by induction on the length of~$tr$.  Hence if~$P$ and~$Q$
are similar, they are equivalent in the Availability Sets Traces Model, and
hence all our other models.  Simulation is strictly finer than our models,
since it distinguishes $a \then b \then c \then STOP \extchoice a \then b
\then d \then STOP$ and $a \then (b \then c \then STOP \extchoice b \then d
\then STOP)$, for example. 

A further possible class of models that we hope to investigate would record
events that were available as alternatives to the events that were actually
performed, and that were available from the \emph{same state} as the events
that were performed.  For example, such a model would distinguish
\begin{eqnarray*}
P & = & a \then c \then STOP \extchoice b \then STOP \\
Q & = & (a \then STOP \extchoice b \then STOP) \timeout  a \then c \then STOP,
\end{eqnarray*}
since $P$ can perform~$\trace{a,c}$, with $b$ available from the state where
$a$ was performed; but $Q$ does not have such a behaviour.  Note that these
two processes are equivalent in all the other availability models in this
paper.
%

Two further possible directions in which this work could be extended would be
(A)~to record what events are \emph{not} available, or (B)~to record the
\emph{complete} set of events that are available.  We see considerable
difficulties in producing such models.  To see why, consider the process $a
\then P$.  There are two different ways of viewing this process (which amount
to different operational semantics for this process):
\begin{itemize}
\item 
One view is that the event $a$ becomes availability \emph{immediately}.  With
this view: in model~A, one cannot initially observe the unavailability of~$a$;
in model~B, the initial complete availability set is~$\set{a}$.  However,
under this view, the fixed point theory does not work as required, since
$\Div$ is not the bottom element of the subset ordering: in model~A, $\Div$
has $a$ initially unavailable; in model~B, $\Div$'s initial complete
availability set is~$\set{}$; these are both behaviours not exhibited by~$a
\then P$.  Further, under this view, nondeterminism is not idempotent,
since, for example, $a \then P \intchoice a \then P$ has $a$ unavailable
initially; one consequence is that the proof of the no-junk result cannot be
easily adapted to this view.

\item
The other view is that $a \then P$ takes some time to make the event~$a$
available: initially $a$ is unavailable, but an internal state change occurs
to make~$a$ available.  With this view: in model~A, one can initially observe
the unavailability of~$a$; in model~B, the initial complete availability set
is~$\set{}$.  However, under this view, it turns out that the state of $a
\then P$ after the $a$ has become available cannot be expressed in the syntax
of the language; this means that the proof of the no-junk result cannot be
easily adapted to this view.  (Proving a full abstraction result is
straightforward, though.)
\end{itemize}

As noted in the introduction, in~\cite{gavin:ready} we considered models for
an extended version of CSP with a construct ``$\If \Ready a \Then P \Else
Q$''.  This construct tests whether or not its environment offers $a$, so the
model has much in common with model~A above (and was built following the
second view).  As such, it did not have a no-junk result.  Further, it did not
have a full abstraction result, since it distinguished $\If \Ready a \Then P
\Else P$ and~$P$, but no reasonable test would distinguish these processes.


\paragraph{Acknowledgements}
I would like to thank Bill Roscoe, Tom Gibson-Robinson and the anonymous
referees for useful comments on this work.

\bibliographystyle{eptcs}
\bibliography{avail}


\end{document}